\documentclass{article}

\usepackage[utf8]{inputenc} 
\usepackage[T1]{fontenc}    
\usepackage{hyperref}       
\usepackage{url}            
\usepackage{booktabs}       
\usepackage{amsfonts}       
\usepackage{nicefrac}       
\usepackage{microtype}      
\usepackage{xcolor}         

\usepackage{amsthm}
\usepackage{amsmath}
\usepackage{amssymb}
\usepackage{fullpage}
\usepackage{natbib}
\usepackage{amsthm}

\newcommand{\eps}{\varepsilon}
\usepackage{xcolor}

\renewcommand{\flat}{\operatorname{flat}}
\newcommand{\polylog}{\operatorname{polylog}}
\newcommand{\R}{\mathbb{R}}
\renewcommand{\S}{\mathbb{S}}
\newcommand{\E}{\mathbb{E}}
\newtheorem{theorem}{Theorem}
\newtheorem{corollary}{Corollary}
\newtheorem{fact}{Fact}
\DeclareMathOperator{\Attn}{Attn}

\title{Nearly Optimal Attention Coresets}
\author{Edo Liberty, Alexandr Andoni, Eldar Kleiner}
\date{\nonumber}

\begin{document}
\maketitle

\begin{abstract}
We consider the problem of estimating the Attention mechanism in small space, and prove the existence of coresets for it of nearly optimal size. 
Specifically, we show that for any set of unit-norm keys and values $(K,V)$ in $\R^d$, 
there exists a subset $(K',V')$ of size at most $O({\sqrt{d} e^{\rho+o(\rho)}/\epsilon})$ such that 
\[
\left\| \Attn(q,K,V)- \Attn(q,K',V') \right\| \le \eps
\]
simultaneously for all queries whose norm is bounded by $\rho$. 
This outperforms the best known results for this problem. 
We also offer an improved lower bound showing that $\eps$-coresets must have size $\Omega({\sqrt{d} e^{\rho}/\epsilon})$.

\end{abstract}
\section{Introduction}
Modern transformer-based generative models~\cite{vaswani2017attention,achiam2023gpt,Claude,reid2024gemini}
make heavy use of the attention mechanism. 
Each attention head typically retains a large KV-cache (set of keys and values), that serves as the working memory and context of the model during generation~\cite{shazeer2019mqa,pope2023efficiently}.
As context windows grew, the size of this cache become one of the dominant bottlenecks of deploying large language models~\cite{kaplan2020scaling,reid2024gemini,fu2024data,kwon2023efficient,sheng2023flexgen}.
Compressing KV-caches is therefore a central problem for the practical deployment of long-context models.
There has already been substantial effort in this direction (see Section~\ref{sec:related} for further related work).

We study the question of optimal KV-cache compression with error guarantees.
A powerful general technique for compression is {\em coreset} selection. 
It is widely used for compressing large datasets while provably preserving quality for a rich families of queries~\cite{karnin2019discrepancy,PT20,CKW24,GLPW16}.
These use deep existential results in discrepancy theory which are recent algorithmic breakthroughs~\cite{B98,B12,B10,DNTT18,ALS21,BJSS19,KRR23}, and these techniques are now known to yield essentially optimal compression rates for several related objectives.

These include all sums of scalar functions of the dot product and Euclidian distance~\cite{karnin2019discrepancy}. More recently, this approach was also applied to streaming attention approximation~\cite{kochetkova2025streaming}.

Perhaps surprisingly, the optimal size of the coreset for the attention problem has not been known yet, and our goal here is to establish such a bound. Building on the above line of work, we construct a subset of key-value pairs that uniformly approximates the attention output for every query of bounded norm, with a coreset size that nearly matches our information-theoretic lower bound.

\paragraph{Problem definition.}
In the attention computation problem, we are given a set of  $n$ keys and values ($k_i,v_i\in \R^d$) denoted by $K,V\in \R^{n\times d}$.
Without loss of generality, we assume keys are centered and both keys and values are at most norm one. 
For a query vector $q\in \R^d$, and we need to output the ``attention vector'' $\Attn(q,K,V)=A(q,K,V)/B(q,K)$, where
\[
A(q,K,V) = \sum_i e^{q^T k_i}v_i \mbox{ \;\; and \;\; } B(q, K) = \sum_i e^{q^T k_i}.
\]
To simplify the presentation we remove a customary (but superfluous) $1/\sqrt{d}$ factor from the exponents (see Section~\ref{notation} for precise setup).

A {\em coreset} is a subset $K',V'$ of $K,V$ respectively, of size $S$ such that, for some error tolerance $\eps>0$, and for any query $q\in \R^d$ (or some more restricted class):
\begin{equation}\label{eqn:AttnApprox}
\| \Attn(q,K,V) -  \Attn(q,K',V') \|_2 \le \eps.
\end{equation}

\paragraph{Our contributions.} Our main result is the following upper bound on the smallest coreset for approximating $\Attn(q, K,V)$. 

\begin{theorem}[See full Theorem~\ref{thm:main}]\label{thm:mainUB}
For any $\rho,\eps>0$, any $(K,V)$ of unit-norm vectors, there exist subsets $(K',V')$ of size $S=O(\frac{\sqrt{d}}{\eps} \cdot e^{\zeta})$ 
such that Eqn.~\eqref{eqn:AttnApprox} is satisfied simultaneously for all $q \in \R^d$ of norm at most $\rho$.
\end{theorem}

We also prove a matching lower bound on the coreset size:
\begin{theorem}[See full Theorem~\ref{thm:coresetMainLB}]\label{thm:mainLB}
Fix any $\rho,\eps>0$. Any coreset for $(K,V)$ cache of unit-norm vectors, satisfying Eqn.~\eqref{eqn:AttnApprox} for queries of norm at most $\rho$, must have size $S=\Omega(\frac{\sqrt{d}}{\eps} \cdot e^{\rho})$.
\end{theorem}

\paragraph{Limits of compression and on the need to bound the norm of $q$.} Without any constraints on or relaxations of the attention problem, no compression is possible.  
The attention mechanism approximates retrieval; the very terminology of ``keys'' and ``values'' is borrowed from associative lookup.
Concretely, attention can be made to reproduce any single value vector almost exactly: place a query in the direction of the corresponding key and scale its norm to make that key dominate the softmax.
Any sketch that supports arbitrary queries must therefore essentially store the full set of values, which seems to rule out any lossy compression.

The escape from this impossibility is to bound the query norm.
A query that ``cherry-picks'' a specific value is necessarily one of large norm: the larger the desired separation in softmax weights, the larger the query norm required.
We therefore seek to compress the cache so that attention is approximated for \emph{all} queries below a fixed norm threshold $\rho$. 
Without such a constraint, the lower bounds such as \cite{kochetkova2025streaming} or our improved Theorem~\ref{thm:mainLB} rule out any compression. 

\paragraph{Discussion of the approach to Theorem~\ref{thm:mainUB}.} Our main technical contribution is to design a coreset for the $A(q, K,V)$ vector;
optimal coresets for $B(q, K)$ are already known and are a direct outcome of \cite{karnin2019discrepancy}.

The result in \cite{karnin2019discrepancy} produces small corsets for general sums of all analytic functions of the shape $\sum_i f(q^Tk_i)$ or $\sum_i  f(\|q-k_i\|)$.
It relies heavily on deep result in discrepancy theory called Banaszczyk's vector balancing theorem \cite{B98}. 
The result in \cite{karnin2019discrepancy} does not naturally extend to vector valued functions (like $A(q,K,V)$). Indeed, such an extension is the main focus of this paper. 

Very recent work gave a streaming coreset contraction \cite{kochetkova2025streaming} based on discrepancy minimization that achieves size $S$ roughly $\tfrac{\sqrt{d}\log n}{\eps} e^{2\rho}$ (when mapped to our setup).
They design an efficient algorithm by applying a breakthrough algorithm for online discrepancy minimization \cite{ALS21}.
They use a linearization of the exponential kernel $e^{z_1^Tz_2} = \psi(z_1)^T\psi(z_2)$ and run Algorithm 1 in \cite{ALS21} on $\psi(k_i)$. 
As is common for operating on infinite-dimensional kernel embeddings $\psi$, the authors use the kernel trick to sidestep it.
The crux of their proof uses Banaszczyk's theorem to show that $\sum_i \psi(k_i)$ is contained in an appropriately small ball. 

In this paper, we surgically modify (and simplify) the result in \cite{karnin2019discrepancy} to prove the existence of smaller coresets.
We too use Banaszczyk's theorem but in a different way. Specifically, we craft a convex body that more tightly bounds certain tensor norms. 
As a result, this paper improves the best known upper bound by a factor of $\log(n)\cdot e^{\rho - o(\rho)}$. We show our bound is a near-optimal coreset size 
by proving a matching coreset lower bound. 

\subsection{Related work}
\label{sec:related}
Work on compressing the KV-cache has, broadly, taken four directions:
\begin{itemize}
\item{\em Floating point compression:}
The first is to keep all key-value pairs but represent each entry using fewer bits.
These include a large number of techniques which are out of scope for this manuscript. 
Notable examples include~\cite{simhash,jegou2010product,gao2024rabitq,zandieh2025turboquantonlinevectorquantization}.
These methods reduce per-entry storage but leave the number of cached tokens unchanged.

\item{\em Distillation:}
A second line replaces the original keys and values with a different, smaller set that mimics the same attention behavior.
Examples include clustering keys and values to a small set of centroids~\cite{zandieh2024subgen}, learning compact representations of long contexts~\cite{eyuboglu2506cartridges}, prefix and prompt tuning~\cite{li2021prefix}, and attention-matching distillation~\cite{zweiger2026fast}. The substituted set need not be a subset of the original vectors.
The approaches require access to the query distribution to avoid degrading quality. 

\item{\em Linearization:}
A third line abandons the explicit cache altogether and replaces the attention mechanism with a linear-time surrogate based on random or learned feature maps~\cite{rahimi2007randomfeatures,katharopoulos2020transformers,choromanski2021performer,wang2020linformer,schroder2026wildcatnearlinearattentiontheory}, or with state-space and recurrent alternatives~\cite{gu2024mamba,beck2024xlstm}. Such models trade exact softmax attention for fixed-size state but require architectural changes and retraining.

\item{\em Coresets selection:}
A fourth direction, and the one taken in this paper, is \emph{coreset} construction: keep a small subset of the original key-value pairs and discard the rest.
Heuristics in this family include attention-sink retention~\cite{xiao2024streamingllm}, accumulated-attention scoring~\cite{zhang2023h2o,zhang2023h2o,li2024snapkv,liu2024scissorhands,liu2024scissorhands}, layer- and head-adaptive budget allocation~\cite{cai2024pyramidkv,fu2024not,xiao2025duoattention}, query-aware sparsity~\cite{tang2024quest}, model-guided eviction~\cite{ge2024fastgen}, sparse long-context patterns~\cite{beltagy2020longformer}, and short-window memorization~\cite{cabannes2025short}. A recent line gives provable guarantees for coreset-style attention approximation via discrepancy theory~\cite{kochetkova2025streaming,zandieh2024subgen,zandieh23kdeformer,hyperattention}.
\end{itemize}

\paragraph{Computing attention for all queries.}
A related but distinct line studies the cost of computing the full attention matrix-vector product for $n$ queries, rather than compressing a fixed KV-cache for future queries.
Exact implementations such as FlashAttention optimize memory traffic and hardware utilization but still compute the dense attention pattern~\cite{dao2022flashattention}.
On the theory side, subquadratic or near-linear algorithms are possible under additional smoothness, bounded-entry, or approximation assumptions, often via kernel-density estimation, sampling, or kernel data structures~\cite{zandieh23kdeformer,hyperattention,CKW24,schroder2026wildcatnearlinearattentiontheory}.

We highlight the work of~\cite{alman2023fastattentionboundedentries} who parametrize the complexity of the problem as a function of the size of the entries of $K,V,$ and $q$'s, similar to our bound on norm $\rho$.
They
identify a sharp bounded-entry transition: for $d=O(\log n)$ and entries below $\sqrt{\log n}$, attention can be approximated in near-linear time, while entries of order $\sqrt{\log n}$ lead to SETH-based barriers against truly subquadratic algorithms.
Our setting is complementary: we ask how many key-value pairs must be retained so that a compressed cache answers every bounded-norm query, independent of the number of future queries.

\section{Setup and notation}\label{notation}

\noindent {\bf Unit normed values:} We assume that $\|v_i\| \le 1$, without loss of generality. 
This is achieved by scaling $v$ by $\max_i \|v_i\|_2$. The attention vector, and the eventual norm of the error, scales linearly with the same factor. 
The results from \cite{kochetkova2025streaming} obtain bounds that depend on $\sum_i \|v_i\|_2$ which is could be more general. 
They split the $K,V$ set in logarithmic number of bins whose value norms are at most a constant factor away from one another. 
This obviously incurs an additional logarithmic factor in the coreset size. For brevity and simplicity, we do not add this discussion to this paper. 

\noindent {\bf Nearly centered keys:} We assume $\|\sum_i k_i\| \le c\sqrt{d}\log(n)$ without loss of generality. 
Translating keys by any vector does not change the attention expression.
Specifically, subtracting $\frac{1}{n}\sum_i k_i$ from each key yields $\|\sum_i k_i\| = 0$ and minimizes $\sum_i \|k_i\|^2$ which should heuristically work in our favor.
In the worst case, it can could double the norm of some keys.

\noindent {\bf Unit normed keys:} We also assume $\|k_i\| \le 1$ without loss of generality, since the attention is a function of $q^T k_i$.
Hence, dividing all key vectors by $\max_i{\|k_i\|}$ and multiplying queries by the same factor changes nothing.

\noindent {\bf Query norms:} Throughout, $\rho$ denotes the maximal radius $\|q\| \le \rho$ to which the bounds hold. 
The coreset size and discrepancy bounds will depend on a slightly larger radius $\zeta := \rho + \frac{1}{2}\log(\rho) + o(\log\log{\rho})$.

\noindent {\bf No temperature $\sqrt{d}$ rescaling:} The standard notation uses weights $e^{q^Tk_i/\sqrt{d}}$. 
For the same reason as above, we can absorb the $\sqrt{d}$ factor in the exponent into the norms of $q$. This only simplifies the notation.
Therefore, readers who are used to the standard notation should think of our query vectors $q$ as $q \cdot \max_i \|k_i\|/\sqrt{d}$.
As a practical note, $k_i$ are usually initialized to be random Gaussians. 
So $\|k_i\|$ are typically $\sqrt{d}$ and $\max_i \|k_i\|/\sqrt{d}$ can roughly be thought of as a constant. 

\noindent{\bf Tensor notations:}
We use $X = x^{\otimes m}$ to denote $d$-dimensional $m$-tensor of $x$ outer-producted with itself $m$-times $X[i_1,i_2,\ldots,i_m] = x_{i_1} x_{i_2} \cdots x_{i_m}$. 
By convention $x^{\otimes 0} = 1$ is a scalar or a zero-dimensional tensor.
We use $\flat(X)$ to denote a $d^m$ long vector containing the individual entires of $X$. Finally, we use $\langle X, Y\rangle := \flat(X)^T\flat(Y)$.

\noindent{\bf Constants:} We use $c$ to indicate universal constants but reuse the same symbol for different constants.

\section{Statement of main results and proof outline}

We begin by stating the main result from Section~\ref{sec:numerator}. 
It informally says that one can halve the cache size and incur an error that is independent of $n$.
\begin{theorem}[Proof in Section~\ref{sec:numerator}]\label{thm:A}
For any $(K,V)$ there exist subsets $(K',V')$ of at most half its size 
such that simultaneously for all $\|q\| \le \rho$:
\begin{equation}
\|A(q,K,V) - 2A(q,K',V')\| \le c\sqrt{d}\cdot e^{\zeta}.
\end{equation}
\end{theorem}
We remind the reader that here, and throughout $\zeta := \rho + \frac{1}{2}\log(\rho) + o(\log\log{\rho})$.
This is also true for $B(q,K,V)$. 
One can see that by using Theorem~\ref{thm:A} and setting all $v_i$ to $[1,0,0,...]$. But, it is also a direct outcome of \cite{karnin2019discrepancy}.

\begin{corollary}\label{cor:B}
For any $(K,V)$ there exist subsets $(K',V')$ of at most half its size 
such that simultaneously for all $\|q\| \le \rho$
\begin{equation}
\|B(q,K) - 2B(q,K')\| \le c\sqrt{d}\cdot e^{\zeta}.
\end{equation}
\end{corollary} 

\begin{fact}[Proof in Section~\ref{sec:simbal}]\label{k_balance}
For any set of unit normed keys $K$ there exist a subset $K'$ of at most half its size such that 
\begin{equation}
\left\|\sum_{k_i\in K} k_i - 2\sum_{k'_i\in K'} k'_i \right\| \le c\sqrt{d}.
\end{equation}
\end{fact}

For the proof to hold, we need all three conditions at the same time. Section~\ref{sec:simbal} and Theorem~\ref{thm:simbal} show this is indeed possible:

\begin{fact}[Proof in Section~\ref{sec:simbal}]\label{fact:simbal}
For any $(K,V)$ there exist subsets $(K',V')$ of at most half its size such that all three conditions hold simultaneously. 
\end{fact} 

Using all three conditions above we are nearly ready to derive the main result. 
The missing ingredient to show that $B(q,K,V)$ is not too small.
Assume $\|\sum_i k_i\| \le c\sqrt{d}\log(n)$.
By Jensen and convexity $\sum_i e^{q^T k_i} \ge n e^{q^T (\frac{1}{n}\sum_i k_i)} \ge ne^{-\|q\|\|\sum_i k_i\|/n} \ge n(1 - c \rho \sqrt{d}\log(n)/n) \ge n/2$ 
for $n \ge 2c \rho \sqrt{d}\log(n)$. In other words, as long as the keys are ``roughly centered" and there aren't too few of them, $B(q,K,V) \ge n/2$.

We are finally ready to combine the ingredients above to prove the following: 

\begin{theorem}\label{thm:main}
For any $(K,V)$ there exist subsets $(K',V')$ of size $O(\frac{\sqrt{d}}{\eps} \cdot e^{\zeta})$ 
such that simultaneously for all $\|q\|\le \rho$
\begin{equation}
\| \Attn(q,K,V) -  \Attn(q,K',V') \| \le \eps.
\end{equation}
\end{theorem}
\begin{proof}
For brevity, denote $A := A(q,K,V)$, $A' := A(q,K',V')$ and analogously $B$ and $B'$. 
Let $\delta_A = 2A' - A$ and $\delta_B = 2B' - B$.
From the above we have that $\|\delta_A\| \le c\sqrt{d}e^\zeta$ and $|\delta_B| \le c\sqrt{d}e^\zeta$
\begin{eqnarray}
\left\| \frac{A'}{B'} - \frac{A}{B}\right\| &=& \left\|\frac{2A'}{2B'} - \frac{A}{B}\right\|  = \left\|\frac{A + \delta_A}{B + \delta_B} - \frac{A}{B}\right\| \\
&=& \left\|\frac{\delta_A B - \delta_B A}{B(B + \delta_B)}\right\| \le \frac{\| \delta_A \|}{B - |\delta_B|} + \left\|\frac{A}{B}\right\|\frac{|\delta_B|}{B - |\delta_B|} \\
&\le& \frac{2\|\delta_A\| + 2|\delta_B|}{B}  \le \frac{4c\sqrt{d}e^{\zeta}}{n}.
\end{eqnarray}

Above we used several facts. First $\|A/B\| \le 1$, this holds for any unit norm values.
Second, $B \ge n/2$. This is true in the first step because $\sum_i k_i = 0$. 
But, in general, we can repeat this compression step $\ell \le \log(n)$ times and due to Fact~\ref{k_balance} at all times $\left\|\sum_i k_i\right\| \le c\sqrt{d}\log(n)$.
Finally, we need that $\delta_B \le B/2$ which holds as long as $n \ge c\sqrt{d}e^\zeta$.

This is just the effect of one compression though. Let's iterate. 
Note that at every step the value of $n$ is cut (at least) in half.
We should stop when the added error is at most $\eps/2$. 
In other words, compress $\ell = \log(\frac{\eps n}{4c\sqrt{d}e^{\zeta}})$ times.
\[
\left\| \frac{A'}{B'} - \frac{A}{B}\right\| \le  \frac{4c\sqrt{d}e^{\zeta}}{n} + \frac{4c\sqrt{d}e^{\zeta}}{n/2} + \frac{4c\sqrt{d}e^{\zeta}}{n/4} + \ldots + \eps/2 \le \eps
\]
This completes the proof and provides a coreset of size at most $n/ 2^{\ell} = O(\frac{\sqrt{d}e^{\zeta}}{\eps})$ such that $
\|\Attn(q,k,v) - \Attn(q,k',v')\| \le \eps.
$
\end{proof}

\paragraph{Comparison to Streaming Attention Approximation
via Discrepancy Theory.}
For comparison to \cite{kochetkova2025streaming} we translate their result to our notation.
As a corollary, applying the above $\ell = \log_2(\frac{\eps n}{4c\sqrt{d}e^{\zeta}})$ times shows that:
\begin{corollary}
For any $(K,V)$ there exist subsets $(K',V')$ of size $O(\frac{\sqrt{d}}{\eps} \cdot e^{\zeta})$ and a constant $c$ 
such that simultaneously for all $\|q\| \le \rho$
\begin{equation}
\|A(q,k,v) - c \cdot A(q,k',v')\| \le \eps n.
\end{equation}

\paragraph{Equation 2 and Theorem 3.1 from \cite{kochetkova2025streaming}}
For any $(K,V)$ there exist subsets $(K',V')$ of size $O(\polylog(n)\frac{\sqrt{d}}{\eps} \cdot e^{2\rho})$ and a constant $c$ such that
simultaneously for all $\|q\| \le \rho$
\begin{equation}
\|A(q,K,V) -c \cdot A(q,K',V')\| \le \eps \sqrt{\sum e^{2q^T k_i}} \cdot \sqrt{\sum \|v_i\|^2}.
\end{equation}

For comparison assume our setting in which $\|v_i\| = 1$ and $\sum k_i = 0$. 
Then $\sum_i \|v_i\|^2 = n$. By Jensen and convexity $\sum_i e^{2q^T k_i} \ge n$. 
So, $\eps \sqrt{\sum (e^{q^T k_i})^2} \cdot \sqrt{\sum \|v_i\|^2} \ge \eps n$. 
The sketch size required is larger by factor of $\polylog(n)e^{\rho\cdot(1-o(1))}$ than our upper bound.
\end{corollary}

\section{Bounding the numerator $A(q,K,V)$ approximation}\label{sec:numerator}
In what follows we bound the error incurred by compressing $A(q,K,V) = \sum_i e^{q^T k_i} v_i^T$, proving Theorem~\ref{thm:A}.

\paragraph{The halving trick.}
There is a standard ``trick" to halve the number of items.
Create the error function $E(q,K,V) := \sum_i \sigma_i e^{q^T k_i} v_i^T$ 
by multiplying every summand above with a sign $\sigma_i \in \{-1,1\}$.
Let $(K',V')$ be the set associated with indexes whose $\sigma_i=1$ . Then:
\begin{eqnarray*}
A(q,K,V) + E(q,K,V) =  2A(q,K',V').
\end{eqnarray*}
Note that $(K',V')$ is at most half the size of $(K,V)$. If that's not the case, flip the signs of $\sigma_i$.
At this point, we can now pick arbitrary signs $\sigma$ specifically to minimize $\|E(q)\|$.

\paragraph{Discrepancy argument.}
Here, we use a discrepancy argument to show that there is always some signs $\sigma$ for which $\|E(q)\|$ is small. 
We use a technique similar to \cite{karnin2019discrepancy,kochetkova2025streaming}.
We surgically modify and simplify their arguments to work specifically for this case.

Below we use the Taylor expansion $e^x = \sum_{i=0}^{\infty}x^m/m!$, 
the fact that $||v|| = \max_{u\in\S^{d-1}} v^Tu$, 
that $(q^Tk)^m = \langle q^{\otimes m},k^{\otimes m}\rangle$, and $q^{\otimes m} = \|q\|^m \cdot \hat{q}^{\otimes m}$ where $\hat{q} = q/\|q\|$.
We also define $z(m) := \sqrt{3(m+2)}\log(m+2)$ and $S_m := \sum_i \sigma_i  \frac{k_i^{\otimes m} \otimes v_i}{z(m)}$.

\begin{eqnarray*}
\|E(q)\| &=& \max_{u,q}  \sum_i \sigma_i e^{q^T k_i} v_i^Tu  \\
&=& \max_{q,u} \sum_i \sum_{m=0}^{\infty} \frac{1}{m!} (q^Tk_i)^m v_i^Tu  \\
&=& \max_{q,u} \sum_i \sum_{m=0}^{\infty} \frac{1}{m!} \left\langle q^{\otimes m},  k_i^{\otimes m} \right\rangle v_i^Tu  \\
&=& \max_{q,u}  \sum_{m=0}^{\infty} \frac{1}{m!} \left\langle q^{\otimes m} \otimes u, \sum_i \sigma_i  k_i^{\otimes m} \otimes v_i \right\rangle  \\
&=& \sum_{m=0}^{\infty} \frac{z(m)\|q\|^m}{m!} \max_{\hat{q},u}   \left\langle \hat{q}^{\otimes m} \otimes u, \sum_i \sigma_i  \frac{k_i^{\otimes m} \otimes v_i}{z(m)} \right\rangle \\
&\le& \sum_{m=0}^{\infty} \frac{z(m)\|q\|^m}{m!} \|S_m\|_{\star}
\end{eqnarray*}
Note that $z(m)$ is introduced to multiply and divide the $m$'th summand; the reason for that will become apparent shortly.
We prove below that we can always find $\sigma$ such that for all $m$ simultaneously:
\[
\|S_m\|_{\star} := \max_{\hat{q},u}   \left\langle \hat{q}^{\otimes m} \otimes u, S_m \right\rangle \le c\sqrt{d\log(m)}.
\]
For now, let us first complete the proof assuming this is correct:
\begin{eqnarray*}
\|E(q)\| &\le& c\sqrt{d}\sum_m \frac{z(m) \rho^m \sqrt{\log(m)}}{m!}  \\
&\le& c\sqrt{d} \sum_m \frac{\sqrt{m+2}\log^{3/2}(m+2) \rho^m}{m!} \\
&\le& c\sqrt{d}\cdot e^{\rho+\frac{1}{2}\log(\rho) + O(\log\log(\rho))} = c\sqrt{d}\cdot e^{\zeta}
 \end{eqnarray*}
 Here we use $\|q\| \le \rho$ and $\rho+\frac{1}{2}\log(\rho) + O(\log\log(\rho)) = \zeta$.
This completes the proof.

\textbf{Tensor norm balancing argument.}
We now show that $\|S_m\| \le c\sqrt{dm\log(m)}$ is possible. 
We chose $z(m)$ specifically so that we can apply Banaszczyk's vector balancing theorem  \cite{B98} to the following vectors
\[
x_i = \left[\frac{v_i}{z(0)}  , \frac{\flat( k_i \otimes v_i)}{z(1)} ,\ldots, \frac{\flat(k_i^{\otimes m} \otimes v_i)}{z(m)} , \ldots\right] \ .
\]
For Banaszczyk's theorem to hold we should have $\|x_i\| \le 1$. Our choice of $z(m)$ enforces that 
\[
\|x_i\|^2 = \sum_{m=0}^{\infty} (1/z(m))^2 = \frac{1}{3}\sum_{m=2}^{\infty} \frac{1}{m\log^{2}(m)} \le 1
\]
Moreover, because $S_m := \sum_i \sigma_i  \frac{k_i^{\otimes m} \otimes v_i}{z(m)}$
\[
\sum_i \sigma_i x_i = \left[\flat(S_0)  , \flat(S_1) ,\ldots, \flat(S_m) , \ldots\right] \ .
\]

Banaszczyk's vector balancing theorem to shows there exist signs $\sigma$ such that $\sum_i \sigma_i x_i \in 5\mathcal{K}$ where $\mathcal{K}$ is any convex set of Gaussian Measure at least $1/2$.

We will define convex set $\mathcal{K}$ as follows. 
Let vector $g = [\flat(G_0),\flat(G_1),\flat(G_2),\ldots]$, where $G_0$ is a $d$ dimensional vector, $G_1$ is a $d \times d$ matrix, $G_2$ is $d^3$ 3-tensor, etc.  
\[
\mathcal{K} = \{ g \;\; | \; \forall_m \;\; \|G_m\|_{\star} \le c\sqrt{d\log(m)} \}
\]
Based on Theorem 1 in \cite{tomioka2014spectralnormrandomtensors} 
There exist a constant $c$ such that $\Pr[\|G_m\| \ge c\sqrt{dm\log(m)}] \le 1/20m^2$. 
Note, however, that $\|G_m\|_{\star} \le \|G_m\|$. 
To bound the standard tensor norm, Theorem 1 in \cite{tomioka2014spectralnormrandomtensors} goes through a union bound over an epsilon-net argument over $m$ $d$-dimensional test vectors. 
Here, $\|G_m\|_{\star}$ is constrained to test tensors of the shape  $\hat{q}^{\otimes m}\otimes u$. 
This means the same epsilon-net argument can be made over only two $d$-dimensional test vectors (namely $\hat{q}$ and $u$).
As a result there exist a constant $c$ such that $\Pr[\|G_m\|_{\star} \ge c\sqrt{d\log(m)}] \le 1/20m^2$. 

Union bounding for all $m$ gives that a random Gaussian vector $g$ belongs to $\mathcal{K}$ with probability at least $9/10$.
Therefore, there exist $\sigma$ such that $\sum_i \sigma_i x_i \in 5\mathcal{K}$ and for all $m$ we have $\|S_m\|_{\star} \le 5c\sqrt{d\log(m)}$.

\section{Simultaneous Balancing}\label{sec:simbal}
Here we prove Facts~\ref{k_balance},~\ref{fact:simbal}.

\textbf{Proof of Fact~\ref{k_balance}.}
This is straight forward because $\|k_i\| \le 1$. Banaszczyk's vector balancing theorem applies directly. 
We can define the convex set $\mathcal{K}$ for the Gaussians as the ball of radius $\sqrt{10d}$. 
Applying Markov to the squared norm of $g$, a vector of $d$ Gaussians, we get $\Pr[\|g\|^2 \ge 10d] \le \E[\sum g_i^2] / 10d \le 1/10$. 
There is always a set $K'$ of at most half the key such that $\|\sum_i \sigma_i k_i\|  = \|\sum_i k_i - 2\sum_i k'_i \| \le 5\sqrt{10d}$.

\textbf{Proof of Corrolary~\ref{cor:B}.}
While a coreset for $B(q,K)$ (Corollary~\ref{cor:B}) is already known, it is actually helpful (and trivial) to repurpose our proof above.  
Set $V$ such that all $v_i$ are the vector $[1,0,0,\ldots]$. The value of $B(q,K)$ is simply the first coordinate of $A(q,K,V)$. The same guarantees hold. 

For the proof above to work, however, we need to achieve all three objectives {\it simultaneously}.  
\begin{theorem}[simultaneous balancing]\label{thm:simbal}
Let there be $s$ different sequences of $n$ unit vector $\|x_i^j\| \le 1$. Let the Gaussian measure of $s$ corresponding convex bodies $\mathcal{K}^j$ be at least $1 - \delta_j$. If $\sum_j \delta_j < 1/2$ then there exist signs $\sigma_i$ such that $\sum_i \sigma_i x_i^s \in \sqrt{s}5\mathcal{K}^s$ for all $j \in [s]$.
\end{theorem}
\begin{proof}
This is a trivial extension of Banaszczyk's vector balancing theorem. 
Let $x_i = \frac{1}{\sqrt{s}}[x_i^1,x_i^2,\ldots,x_i^s]$ be a concatenation of $s$ vectors. Clearly $\|x_i\| \le 1$.
Since $\sum_j \delta_j \le 1/2$, by the union bound, the gaussian measure of $\mathcal{K} = \cap_i \mathcal{K}^j$ is at least $1/2$.
So, there exist $\sigma$ such that $\sum_i \sigma_i x_i \in 5\mathcal{K}$. 
Since $\mathcal{K}^j$ are dimension disjoint, this also means that, for all $j \in [s]$, $\sum_i  \sigma_i \frac{1}{\sqrt{s}} x_i^j \in 5\mathcal{K}^j$.
\end{proof}
In our case, use $s=3$, one for $A(q,K,V)$, one for $B(q,K)$ and one for balancing $K$. From the proofs above, for all three, $\delta_j < 1/10$.
Applying the above, we can achieve all three objectives with a single set of signs $\sigma_i$ at the cost of increasing the constant by a factor of $\sqrt{3}$. 

\section{Lower Bounds}
We now prove the following two lower bounds.
Essentially, the first theorem proves a bound of $\Omega(de^{\rho}/\eps)$ {\em bits} required to store a sketch for the Attention computation. The second theorem proves a bound of $\Omega(\sqrt{d}e^\rho/\eps)$ on the size of the coreset (Theorem~\ref{thm:mainLB} from the introduction).

We use the following standard spherical code fact \cite{conway1999sphere}: there is an absolute constant $a>0$ such that, for every $0<\eta\le 1/2$ and every $m\le \exp(a\eta^2 d_k)$, there are unit vectors $u_1,\ldots,u_m\in\R^{d_k}$ with $\langle u_i,u_h\rangle\le \eta$ for every $i\ne h$.

\begin{theorem}
Fix dimensions $d,d_k>1$, $\rho\ge 2$, let $L=e^\rho$, and let $0<\eps\le \eps_0$ for a sufficiently small absolute constant $\eps_0$. Suppose there is a sketching algorithm that, for any input $K\subset \R^{d_k},V\subset\R^d$ with $\|k_i\|_2,\|v_i\|_2\le 1$, produces a sketch $\Attn'(\cdot)$ of size $S$ so that for any $q\in \R^{d_k}$ of norm at most $\rho$, with constant probability,
\[
\|\Attn(q,K,V)-\Attn'(q)\|_2\le \eps .
\]
Then
\[
S=\Omega\!\left(d\cdot \min\{L/\eps,\; L^2,\; \exp(a d_k/\log^2 L)\}\right).
\]
In particular, it is enough to take $d_k=\Omega(\log^2 L\cdot\log(L/\eps))$ for the key dimension not to be the bottleneck.
\end{theorem}

\begin{proof}
We reduce from one-way \textsc{Indexing}. By repeating the sketch a constant number of times, we may assume its success probability is at least $0.99$, changing $S$ by only a constant factor. Set $\eta=1/\log L$. Since $\rho=\log L\ge 2$, we have $\eta\le 1/2$. Let
\[
m=\left\lfloor c\min\{L/\eps,\; L^2,\; \exp(a d_k/\log^2 L)\}\right\rfloor
\]
for a sufficiently small absolute constant $c$, and consider an \textsc{Index} instance $x\in\{0,1\}^{md}$. Let $u_1,\ldots,u_m$ be the codewords from the coding fact above. We identify a bit position with a pair $(i,j)\in[m]\times[d]$. Using public randomness, Alice and Bob draw signs $b_{ij}\in\{\pm1\}$. Alice constructs $m$ key-value pairs
\[
k_i=u_i\in\R^{d_k},
\qquad
v_i=\frac{1}{\sqrt d}(b_{i1}x_{i1},\ldots,b_{id}x_{id})\in\R^d .
\]
Clearly $\|v_i\|_2\le 1$. Alice feeds these pairs to the sketching algorithm and sends the resulting sketch to Bob.

To recover bit $(i,j)$, Bob asks the query $q=\rho u_i$, whose norm is $\rho$. Write
$
w_{ih}=\exp(\rho\langle u_i,u_h\rangle).$
Then $w_{ii}=L$ and $w_{ih}\le e$ for $h\ne i$. Let
$
B=\sum_{h=1}^m w_{ih}
$
be the softmax denominator for this query. The $j$th coordinate of the true attention, after multiplying by the public sign $b_{ij}$, is
\[
b_{ij}\Attn(q,K,V)_j
=
\frac{Lx_{ij}/\sqrt d + N_{ij}}{B},
\qquad
N_{ij}:=\frac{1}{\sqrt d}\sum_{h\ne i} w_{ih}b_{ij}b_{hj}x_{hj}.
\]
The first term is the signal. The second term is noise from all other keys. Over the public signs, $\mathbb{E}N_{ij}=0$ and $\mathbb{E}N_{ij}^2\le e^2m/d$, so Chebyshev gives
$
\Pr\left[|N_{ij}|>L/(10\sqrt d)\right]\le 100e^2m/L^2 .
$
Our choice of $m\le cL^2$ makes this probability a small constant.

On the event $|N_{ij}|\le L/(10\sqrt d)$, the two cases are separated by a gap: if $x_{ij}=1$ then
\[
b_{ij}\Attn(q,K,V)_j \ge \frac{9L}{10\sqrt d\,B},
\]
whereas if $x_{ij}=0$ then
\[
\left|b_{ij}\Attn(q,K,V)_j\right|\le \frac{L}{10\sqrt d\,B}.
\]
Thus thresholding at $L/(2\sqrt d\,B)$ recovers the bit as long as the sketch error in coordinate $j$ is at most $L/(5\sqrt d\,B)$.

It remains to check that the promised $\ell_2$ error gives such a coordinate bound for a random \textsc{Index} query. Conditioned on the sketch succeeding, if $(i,j)$ is uniform in $[m]\times[d]$, then $j$ is uniform in $[d]$, and Markov's inequality gives
\[
\Pr_j\left[
\left(\Attn'(q)_j-\Attn(q,K,V)_j\right)^2>(10\eps)^2/d
\right]\le 1/100 .
\]
Finally, since $B\le L+em$, $m\le cL/\eps$, $\eps\le\eps_0$, and $c,\eps_0$ are sufficiently small, we have $B\le L/(50\eps)$, and hence 
$10\eps/\sqrt d \le L/(5\sqrt d\,B).$
Therefore Bob recovers a uniformly random bit of $x$ with constant probability. A sketch of size $o(md)$ would give a one-way protocol for \textsc{Index} on $md$ bits using $o(md)$ communication, contradicting the standard $\Omega(md)$ lower bound. Hence $S=\Omega(md)$, which is the claimed bound.
\end{proof}

\begin{theorem}
\label{thm:coresetMainLB}
Fix dimensions $d,d_k>1$, $\rho\ge 2$, let $L=e^\rho$, and let $0<\eps\le \eps_0/\sqrt d$ for a sufficiently small absolute constant $\eps_0$. Suppose there is a coreset algorithm that, for any input $K\subset \R^{d_k},V\subset\R^d$ with $\|k_i\|_2,\|v_i\|_2\le 1$, produces a subset $K',V'$ of size $S$, so that for any $q\in \R^{d_k}$ of norm at most $\rho$, with constant probability,
\[
\|\Attn(q,K,V)-\Attn(q,K',V')\|_2\le \eps .
\]
Then
\[
S=\Omega\!\left(d\cdot \min\{L/(\eps\sqrt d),\; L^2,\; \exp(a d_k/\log^2 L)\}\right).
\]
In particular, it is enough to take $d_k=\Omega(\log^2 L\cdot\log(L/(\eps\sqrt d)))$ for the key dimension not to be the bottleneck.
\end{theorem}

\begin{proof}
The reduction is again from one-way \textsc{Indexing}. By constant repetition we may assume the coreset succeeds with probability at least $0.99$. Set $\eta=1/\log L$. Since $\rho=\log L\ge 2$, we have $\eta\le 1/2$. Let
\[
m=\left\lfloor c\min\{L/(\eps\sqrt d),\; L^2,\; \exp(a d_k/\log^2 L)\}\right\rfloor
\]
for a sufficiently small absolute constant $c$, and consider an \textsc{Indexing} instance $x\in\{0,1\}^{md}$. Let $u_1,\ldots,u_m$ be the codewords from the coding fact. We identify a bit position with a pair $(i,j)\in[m]\times[d]$. Using public randomness, Alice and Bob draw signs $b_{ij}\in\{\pm1\}$. Alice constructs $md$ key-value pairs indexed $(i,j)\in [m]\times[d]$:

\[
k_{i,j}=u_i\in\R^{d_k},
\qquad
v_{i,j}=b_{ij}x_{ij}e_j\in\R^d .
\]
Clearly $\|v_{i,j}\|_2\le 1$. Alice runs the coreset algorithm and sends the resulting coreset to Bob.

To recover bit $(i,j)$, Bob asks the query $q=\rho u_i$. As above, write $w_{ih}=\exp(\rho\langle u_i,u_h\rangle)$, so $w_{ii}=L$ and $w_{ih}\le e$ for $h\ne i$. The softmax denominator for the original instance is now
$
B=d\sum_{h=1}^m w_{ih},
$
because there are $d$ pairs for each codeword. After multiplying coordinate $j$ by the public sign $b_{ij}$, the true attention value is
\[
b_{ij}\Attn(q,K,V)_j
=
\frac{Lx_{ij}+N_{ij}}{B},
\qquad
N_{ij}:=\sum_{h\ne i} w_{ih}b_{ij}b_{hj}x_{hj}.
\]
Over the public signs, $\mathbb{E}N_{ij}=0$ and $\mathbb{E}N_{ij}^2\le e^2m$, so
$
\Pr\left[|N_{ij}|>L/10\right]\le 100e^2m/L^2 .
$
Since $m\le cL^2$, this failure probability is a small constant.

On the event $|N_{ij}|\le L/10$, the two bit values are separated: if $x_{ij}=1$ then
\[
b_{ij}\Attn(q,K,V)_j \ge \frac{9L}{10B},
\]
whereas if $x_{ij}=0$ then
\[
\left|b_{ij}\Attn(q,K,V)_j\right|\le \frac{L}{10B}.
\]
Thus thresholding at $L/(2B)$ recovers the bit provided the coreset approximation has coordinate error at most $L/(5B)$.

As before, the $\ell_2$ approximation guarantee implies the desired coordinate guarantee for a random \textsc{Indexing} query. Conditioned on the coreset succeeding, if $(i,j)$ is uniform in $[m]\times[d]$, then $j$ is uniform in $[d]$, so Markov's inequality gives
\[
\Pr_j\left[
\left|\Attn(q,K,V)_j-\Attn(q,K',V')_j\right|^2>(10\eps)^2/d
\right]\le 1/100 .
\]
Finally, since $B\le d(L+em)$, $m\le cL/(\eps\sqrt d)$, and $\eps\le\eps_0/\sqrt d$, for $c,\eps_0$ sufficiently small,
$
B\le \frac{L\sqrt d}{50\eps},
$
and therefore
$
10\eps/\sqrt d\le L/(5B).
$
Bob can therefore recover a uniformly random bit of $x$ with constant probability from the coreset. If the coreset had size $o(md)$, then Alice could send the indices of the selected pairs and corresponding $x_{i,j}$'s (using $o(md)$ retained pairs of information), and Bob could run the above construction and the decoder. This contradicts the standard $\Omega(md)$ lower bound for \textsc{Indexing}. Hence $S=\Omega(md)$, which is the claimed bound.
\end{proof}

\section{Discussion}
In this paper we prove the best possible coreset size is essentially $\tfrac{\sqrt{d}}{\eps}e^\zeta$.
The upper bound is off by a factor of $O(e^{\frac{1}{2}\log(\rho) + O(\log\log(\rho))}) = O(\rho^{1/2 + o(1)})$.
To qualify how ``nearly optimal" our results is, consider the following. 
Based on the lower bound, coreset are only viable when $\zeta \le \log(\eps n/\sqrt{d})$ 
because otherwise $c\sqrt{d}e^\zeta /\eps > n$ which mean our coreset is larger than the data. 
So, in the worst case, the gap between the upper and lower bound is roughly $\sqrt{\log(n)}$.
We conjecture that closing this gap will require completely different approaches. 

\bibliographystyle{plain}
\bibliography{attention_coresets_arxiv}

\end{document}